\documentclass[10pt,journal,twocolumn]{IEEEtran}
\IEEEoverridecommandlockouts
\usepackage{cite}
\usepackage{amsmath,amssymb,amsfonts,booktabs}
\usepackage{algorithmic}
\usepackage{graphicx}
\usepackage{textcomp}
\usepackage{xcolor}
\usepackage{enumerate}
\usepackage{epstopdf}
\usepackage{subcaption}
\usepackage{amsmath,amsfonts}
\usepackage{amsthm}
\usepackage{amssymb}
\usepackage{easyReview}
\usepackage{algorithmic}
\usepackage{algorithm}
\usepackage{array}
\usepackage[caption=false,font=normalsize,labelfont=sf,textfont=sf]{subfig}
\usepackage{stfloats}
\usepackage{url}
\usepackage{bm}
\usepackage{makecell}
\usepackage{verbatim}
\usepackage{caption}
\usepackage{subfig}
\usepackage{subcaption}
\usepackage{cite}
\usepackage{booktabs}
\usepackage{amsthm}

\newtheorem{proposition}{Proposition}

\usepackage{adjustbox}

\usepackage{flushend}

\usepackage{array}
\definecolor{mygray}{RGB}{240,240,240}
\usepackage{colortbl}

\usepackage{algorithm}
\usepackage{algorithmic}
\allowdisplaybreaks[4]

\def\BibTeX{{\rm B\kern-.05em{\sc i\kern-.025em b}\kern-.08em
		T\kern-.1667em\lower.7ex\hbox{E}\kern-.125emX}}
\begin{document}


\title{UAV-Enabled Fluid Antenna Systems for Multi-Target Wireless Sensing over LAWCNs}

\author{
        Xuhui Zhang,~\IEEEmembership{Member,~IEEE},
        Wenchao Liu,~\IEEEmembership{Graduate Student Member,~IEEE},\\
        Chunjie Wang,~\IEEEmembership{Graduate Student Member,~IEEE},
        Jinke Ren,~\IEEEmembership{Member,~IEEE},\\
        Huijun Xing,~\IEEEmembership{Graduate Student Member,~IEEE},
        Shuqiang Wang,~\IEEEmembership{Senior Member,~IEEE},\\ and
        Yanyan Shen,~\IEEEmembership{Member,~IEEE}

\thanks{
Copyright (c) 2025 IEEE. Personal use of this material is permitted.
However, permission to use this material for any other purposes must be
obtained from the IEEE by sending a request to pubs-permissions@ieee.org.
}


\thanks{
Xuhui Zhang is with Shenzhen Institutes of Advanced Technology, Chinese Academy of Sciences, Guangdong 518055, China, also with the Shenzhen Future Network of Intelligence Institute, the School of Science and Engineering, and the Guangdong Provincial Key Laboratory of Future Networks of Intelligence, The Chinese University of Hong Kong, Shenzhen, Guangdong 518172, China (e-mail: xu.hui.zhang@foxmail.com).
}

\thanks{
Wenchao Liu is with the School of Automation and Intelligent Manufacturing, Southern University of Science and Technology, Guangdong 518055, China (e-mail: wc.liu@foxmail.com).
}

\thanks{
Chunjie Wang is with Shenzhen Institutes of Advanced Technology, Chinese Academy of Sciences, Guangdong 518055, China, and also with the University of Chinese Academy of Sciences, Beijing 100049, China (e-mail: cj.wang@siat.ac.cn).
}

\thanks{
Jinke Ren is with the Shenzhen Future Network of Intelligence Institute, the School of Science and Engineering, and the Guangdong Provincial Key Laboratory of Future Networks of Intelligence, The Chinese University of Hong Kong, Shenzhen, Guangdong 518172, China (e-mail: jinkeren@cuhk.edu.cn).
}

\thanks{
Huijun Xing is with the Department of Electrical and Electronic Engineering, Imperial College London, London SW7 2AZ, The United Kingdom (e-mail: huijunxing@link.cuhk.edu.cn).
}

\thanks{
Shuqiang Wang, and Yanyan Shen are with Shenzhen Institutes of Advanced Technology, Chinese Academy of Sciences, Guangdong 518055, China (e-mail: sq.wang@siat.ac.cn; yy.shen@siat.ac.cn).
}

}

\maketitle

\begin{abstract}
Fluid antenna system (FAS) is emerging as a key technology for enhancing spatial flexibility and sensing accuracy in future wireless systems. This paper investigates an uncrewed aerial vehicle (UAV)-enabled FAS for multi-target wireless sensing in low-altitude wireless consumer networks (LAWCNs) for achieving the low-altitude economy (LAE) missions. We formulate an optimization problem aimed at minimizing the average Cramér–Rao bound (CRB) for multiple target estimations. To tackle this non-convex problem, an efficient alternating optimization (AO) algorithm is proposed, which jointly optimizes the UAV trajectory, the antenna position of the transmit fluid antennas (FAs) and the receive FAs, and the transmit beamforming at the UAV. Simulation results demonstrate significant performance improvements in estimation accuracy and sensing reliability compared to conventional schemes, e.g., the fixed position antenna scheme. The proposed system achieves enhanced sensing performance through adaptive trajectory design and beamforming, alongside effective interference suppression via the flexible FAS antenna repositioning, underscoring its practical potential for precision sensing in the UAV-enabled LAWCNs.
\end{abstract}


\section{Introduction}
\IEEEPARstart {I}{n} recent years, with the vigorous development of the low-altitude economy (LAE), airspace resources over the low-altitude wireless consumer networks (LAWCNs) have gradually expanded from traditional civil aviation and military use to multiple civil scenarios such as urban transportation, logistics and distribution, environmental monitoring, and emergency rescue \cite{wu2025low, 10693833, 10879807}. As the core carrier of the LAWCNs, uncrewed aerial vehicle (UAVs), also known as drones, unmanned aerial vehicles, or autonomous
aerial vehicles (AAVs), are promoting the rapid evolution of emerging industries such as intelligent transportation, smart cities, and intelligent manufacturing by virtue of their advantages of flexibility, low cost, and efficient execution \cite{10632079, chen2025ISCC, 10980172}. However, the widespread application of the LAE missions has also brought many challenges. For example, in the LAWCNs, the targets are dense and highly dynamic, and are susceptible to building occlusion and multi-path propagation, making stable and reliable target perception and environmental perception more complex \cite{10955337, liu2025movable}. At the same time, the diversification and complication of the LAE task requirements, such as multi-target sensing, trajectory prediction, and real-time response, have put forward higher requirements for the accuracy, robustness, and coverage of perception systems.

To address these challenges, wireless sensing technology has gradually become an important technology for achieving the LAE missions over the LAWCNs \cite{10430083}. Compared with traditional sensing methods that rely on visions, wireless sensing has advantages such as all-weather operation, anti-interference capability, and long-distance coverage, enabling simultaneous detection and positioning of multiple targets in complex environments \cite{10571011}. Especially when combined with the high mobility of the UAVs, wireless sensing technology can not only dynamically adjust observation positions but also improve sensing accuracy through beamforming and array optimization, thereby providing more efficient solutions for traffic management, public safety, and resource scheduling in the LAWCNs.
However, owing to the high mobility of the UAVs, the channels between the UAVs and the targets can vary rapidly. As a result, the beamforming gain achieved with a fixed-position antenna array may be severely degraded by such channel variations, thereby significantly limiting the performance of wireless sensing.

Thankfully, the emergence of fluid antenna system (FAS) has brought new opportunities for overcoming the channel degradation in the UAV-enabled wireless sensing \cite{9264694}.
{
Unlike traditional fixed position antenna (FPA) arrays, the FAS can dynamically adjust the antenna position of the array elements according to the channel environment, the location of the transmitter and the receiver, and the LAE mission requirements, thereby significantly improving spatial resolution and sensing performance with limited hardware resources through the flexible beamforming over the LAWCNs \cite{10146274, 10146286}.
}
This feature is particularly suitable for the UAV-enabled multi-target sensing tasks: by combining the UAV trajectory design, the beamforming optimization, and the antenna array adaptive configuration, the FAS-aided UAV-enabled systems can effectively overcome the problems of signal fading and multi-path interference in low-altitude environments, achieving high-precision sensing performance of multiple targets in the LAWCNs. 

Although some studies have explored trajectory optimization and beamforming design for the UAV-enabled wireless sensing systems, most existing works are based on the traditional FPA, making it difficult to fully exploit the potential of array reconfigurations. In addition, existing works often focus on single-target or static scenarios, with insufficient consideration given to multi-target sensing problems in dynamic environments. How to simultaneously optimize the trajectory, the beamforming, and the array structure to achieve more efficient multi-target sensing remains a pressing challenge to be addressed.
{Motivated by above, this paper studies a novel UAV-enabled FAS for the multi-target sensing missions over the LAWCNs.}
The main contributions of this paper are summarized as follows:
\begin{itemize}
    \item We investigate a novel UAV-enabled FAS for multi-target wireless sensing, where the UAV is flying over the LAWCN to esitimate the multiple ground targets for specific LAE missions.
    A comprehensive optimization problem is formulated to minimize the average Cramér–Rao Bound (CRB) of the targets.
    \item To address the non-convexity of the formulated problem, we decompose it into three subproblems and propose an efficient alternating optimization (AO)-based algorithm that iteratively optimizes the UAV trajectory design, the transmit and receive fluid antennas (FAs), and the UAV transmit beamforming.
    \item Extensive simulations are carried out to evaluate the proposed scheme with several existing benchmarks. The results confirm significant improvements in sensing accuracy and operational reliability, demonstrating the effectiveness of our joint optimization scheme.
\end{itemize}

\textit{Organizations:}
The rest of this paper is organized as follows. Section II reviews the related works. Section III presents the system model of the UAV-enabled FAS for the multi-target wireless sensing and formulates the CRB optimization problem. In Section IV, the problem is firstly transformed, then solved by an AO-based algorithm. Then, the computational complexity and convergence behavior of the proposed AO-based algorithm are analyzed. The numerical results are provided in Section V, and Section VI concludes the paper and introduces the future directions.

\section{Related Works}
{
This section gives a systematic review of the existing works related to the consumer-oriented UAV systems and the FAS.
The first part reviews the recent advancements in the UAV-enabled LAWCNs for consumer demands, underlining the benefits and application potential for integrating UAV into consumer LAE applications.
The second part reviews the existing works in the UAV-enabled wireless sensing, highlighting their advantages and limitations.
The third part discusses the emerging integration of the FAS with UAV, which introduces new beamforming capabilities through dynamic antenna configuration.}
To clearly contextualize our contribution, Table \ref{trw} summarizes key differences between our work and selected prior studies, particularly in terms of the system model design, the optimization variables, and the performance metrics.

\subsection{{UAV-Enabled LAWCNs}}
{Recent works have increasingly focused on UAV-enabled LAWCNs tailored to meet the specific demands of consumer-oriented LAE applications.
For instance, to process computation-intensive and time-sensitive topological tasks, a UAV-enabled mobile edge computing (MEC) system for edge consumer users was investigated in \cite{11124264}.
For adversarial anomaly detection, a consumer-grade UAV-enabled MEC system was studied in \cite{11192518}, where the detection accuracy was improved in the decentralized federated learning for consumer users.
In a consumer-to-consumer communication system, a collaborative task processing system enabled by multiple UAVs for balancing energy efficiency and computational timeliness was studied in \cite{11222717}.
Moreover, a UAV-enabled consumer ad hoc network was investigated in \cite{11222750}, where dynamic frequency switching and rate adaptation were optimized to mitigate interference leveraging machine learning technique.
For dynamic weather conditions and customer demands, a real-time path design algorithm was proposed for a collaborative truck-UAV delivery system \cite{11205471}.
The integration of UAVs into consumer-oriented LAWCNs holds tremendous promise for shaping the future of the LAE applications through flexible, scalable, and user-centric connectivity.
}

\subsection{UAV-Enabled Wireless Sensing}
{
UAVs offer unique advantages for wireless sensing, such as flexible deployment and improved line-of-sight channel conditions. Several previous studies have begun to explore the UAV-enabled wireless sensing in some LAE applications \cite{10756618, 10680056, 10376413}.
}
For instance, in \cite{10756618}, a full-duplex UAV-enabled wireless sensing network was investigated, aiming to maximize learning performance while ensuring sensing quality.
In \cite{10680056}, a UAV-enabled multi-user tracking and beamforming design over the space-air-ground integrated network was studied, where two methods for estimating the location information of the ground mobile users at the UAV were proposed by the wireless sensing.
Moreover, \cite{10376413} introduced a robust radio frequency-based UAV-enabled sensing approach, in which the received UAV-emitted RF signal was processed to obtain the minimum variance distortionless response spectrum, from which a feature vector is extracted for reliable UAV detection and identification.
{As such, the UAV-enabled consumer-oriented wireless sensing is promising to empower LAE applications by delivering real-time, personalized environmental awareness, e.g., from smart city navigation to immersive location-based services.}
{
However, these previous works are inherently limited by the spatial rigidity of the FPA arrays, which restricts their adaptability to dynamic channel conditions and mobile sensing scenarios.
}

\subsection{UAV-Enabled FAS}
{
In contrast to the FPA arrays, the FAS introduces a new degree of freedom (DoF) by enabling dynamic reconfiguration of antenna positions, thereby enhancing adaptability to varying channel conditions and improving spatial gain. Recently, several preliminary studies have explored the integration of FAS into UAV-enabled LAWCNs, demonstrating potential improvements in spectral efficiency under mobility \cite{11048899, 11078433, 10654366, zhang2025dc, 11148216}.
For instance, the relay UAV was assisted in an FA-enabled communication system to help improve the outage probability \cite{11048899}.
In \cite{11078433}, a dynamic port-reconfigurable channel model was proposed for the air-to-ground communications in the UAV-enabled FAS.
Moreover, the downlink and uplink scenarios in the UAV-enabled FAS were studied \cite{10654366, zhang2025dc}, where the user communication rates were optimized with the help of the FAs.
Besides, the channel estimation for the UAV-enabled FAS was optimized in \cite{11148216}, by maximizing the spatial diversity gain introduced by the FAs.
{Empowered by flexible beamforming, UAV-enabled FAS for LAWCNs can deliver more efficient LAE services to both aerial and ground-based consumers.}
However, existing studies failed to fully exploit the joint design of the UAV mobility and the DoF offered by the FAS to support multi-target sensing with the accurate metrics.
}

\begin{table}[!htbp] \footnotesize
	\centering
	\caption{{Comparison of Selected Existing Works}}
	\label{trw}
	\begin{tabular}
    {
    m{1.2cm}<{\centering}
    m{0.6cm}<{\centering}
    m{0.95cm}<{\centering}
    m{0.95cm}<{\centering}
    m{0.95cm}<{\centering}
    m{0.6cm}<{\centering}
    m{0.6cm}<{\centering}
    }
  \toprule	
  Works & UAV & Trajectory Design & Wireless Sensing
 & Multiple Targets & CRB & FAS \\	
  \midrule  %
  \cite{10756618} & $\surd$ & $\surd$ & $\surd$ & $\times$ & $\times$ & $\times$ \\
  \cite{10680056} & $\surd$ & $\times$ & $\surd$ & $\surd$ & $\times$ & $\times$ \\
    \cite{10376413} & $\surd$ & $\times$ & $\surd$ & $\surd$ & $\times$ & $\times$ \\ 
  \cite{11048899,11078433} & $\surd$ & $\times$ & $\times$ & $\times$ & $\times$ & $\surd$ \\
    \cite{10654366, zhang2025dc} & $\surd$ & $\surd$ & $\times$ & $\times$ & $\times$ & $\surd$ \\
  \cite{11148216} & $\surd$ & $\times$ & $\surd$ & $\surd$ & $\times$ & $\surd$ \\
  \cellcolor{mygray}{\textbf{Ours}} & \cellcolor{mygray}{$\boldsymbol{\surd}$} & \cellcolor{mygray}{$\boldsymbol{\surd}$} & \cellcolor{mygray}{$\boldsymbol{\surd}$} & \cellcolor{mygray}{$\boldsymbol{\surd}$} & \cellcolor{mygray}{$\boldsymbol{\surd}$} & \cellcolor{mygray}{$\boldsymbol{\surd}$}  \\
		\bottomrule
	\end{tabular}
\end{table}

\begin{figure}[ht]
	\centering
	\includegraphics[width=0.85\linewidth]{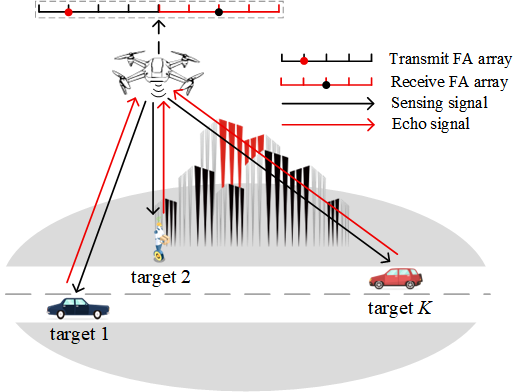}
	\captionsetup{justification=centering}
	\caption{The UAV-enable FAS for multi-target wireless sensing over the LAWCNs.}
	\label{fig:sysmodel}
\end{figure}

\section{System Model and Problem Formulation}
We consider a UAV-enabled wireless sensing system as illustrated in Fig.~\ref{fig:sysmodel}, where the UAV is equipped with the FA array to sense the $K$ targets. The set of the targets is denoted as $\mathcal{K} = \{1,\ldots,K\}$.
The total mission interval is $T$ seconds, which is divided into $N$ equal time slots, with each duration $\tau = T/N$. The set of the time slots is denoted by $\mathcal{N} = \{1,\ldots,N\}$.
We consider a three-dimensional (3D) Cartesian coordinate system. The position of the UAV at time slot $n$ is given by $(\mathbf{q}[n], H)$, where $\mathbf{q}[n] = (x[n], y[n])$ is the horizontal location of the UAV and $H$ denotes the fixed altitude of the UAV.
The position of the $k$-th target is $\mathbf{q}_k = (x_k, y_k)$.

The UAV first transmits the downlink sensing signal towards all targets, then utilizes the reflected echoes to estimate the targets. Hence, the FA array consists of $M_t$ transmit FAs, and $M_r$ receive FAs, where the position of the $M_t$ transmit FAs and $M_r$ receive FAs are denoted by $\mathbf{x}[n] = \{ x_{1}[n], x_{2}[n], \ldots, x_{M_{t}}[n] \}^{\mathsf{T}}$ and $\mathbf{y}[n] = \{ y_{1}[n], y_{2}[n], \ldots, y_{M_{r}}[n]  \}^{\mathsf{T}}$, respectively. To ensure the antenna position feasibility of the FAs, we ensure $x_{1}[n] < x_{2}[n] < \dots < x_{M_{t}}[n], \forall n \in \mathcal{N}$ and $y_{1}[n] < y_{2}[n] < \ldots < y_{M_{r}}[n], \forall n \in \mathcal{N}$.
{To enhance readability, the main notations and their descriptions are given in Table \ref{tnoa}.}

\begin{table}[!htbp] \footnotesize
	\centering
	\caption{{Notations and Descriptions}}
	\label{tnoa}
	\begin{tabular}{c>{\centering\arraybackslash}p{6.825cm}}	
  \toprule	
  Notations& Description\\	
  \midrule  %
    $k,K,\mathcal{K}$ & The index, number, and set of the targets\\ 
  $n,N,\mathcal{N}$ & The index, number, and set of time slots\\
  $\mathbf{q}[n]$ &  The 2D position of the UAV at time slot $n$\\
  $\mathbf{q}_k$ & The 2D location of the $k$-th target\\
  $\mathbf{x}[n]$ & The antenna position of the transmit FA at time slot $n$\\
  $\mathbf{y}[n]$ & The antenna position of the receive FA at time slot $n$\\
  $M_t$ & The number of the transmit FA\\
  $M_r$ & The number of the receive FA\\
  $\bar{n},\bar{N},\bar{\mathcal{N}}$ & The index, number, and set of transmission frames\\
  $\mathbf{s}[n;\bar{n}]$ & The sensing signal of the $\bar{n}$-th frame at time slot $n$\\
  $\mathbf{R}[n]$ & The covariance matrix of the sensing signal $\mathbf{s}[n;\bar{n}]$\\
  $\mathbf{a}_k[n]$ & The steering vectors of the transmit FAs of the UAV for detecting the $k$-th target\\
  $\mathbf{b}_k[n]$ & The steering vectors of the receive FAs of the UAV for detecting the $k$-th target\\
  $\theta_k[n]$ & The vertical AoD from the UAV to the $k$-th target\\
  $\mathbf{W}_k[n]$ & The response matrix of the $k$-th target\\
  $\mathbf{X}_k[n]$ & The transmitted signal from the UAV\\
  $\mathbf{Y}_k[n]$ & The received echo signal at the UAV for estimating the $k$-th target\\
  $\mathbf{N}_k[n]$ & The additive white Gaussian noise\\
  $\mathbf{R}_{\mathbf{x}} [n]$ & The covariance of the transmitted signal\\
  $\boldsymbol{\mathcal{X}}$ & The feasible region of the transmit FAs\\
  $\boldsymbol{\mathcal{Y}}$ & The feasible region of the receive FAs\\
		\bottomrule
	\end{tabular}
\end{table}

\subsection{Sensing Model}

We consider the multi-beam transmission of the sensing signal \cite{lyu2022joint}.
Accordingly, we assume $\bar{N} > M_t$ transmission frames within one time slot, where the set of the frames is denoted as $\bar{\mathcal{N}} = \{ 1,\ldots, \bar{N}\}$. For the $\bar{n}$-th frame, the sensing signal is denoted as $\mathbf{s}[n;\bar{n}]$, which follows a zero-mean circularly symmetric complex Gaussian distribution, and the covariance matrix of $\mathbf{s}[n;\Bar{n}]$ is given by
\begin{equation}
    \mathbf{R}[n] = \mathbb{E} \{ \mathbf{s}[n;\bar{n}] \mathbf{s}^{\mathsf{H}}[n;\bar{n}] \} \succeq \mathbf{0}.
\end{equation}

We consider that the multi-path reflections suffer from the significant attenuation for the LAE mission over the LAWCN \cite{Wang2024Fluid}. Hence, we model the channel model between the UAV and the targets as the single-path propagation model.
Accordingly, the steering vectors of the transmit FAs and the receive FAs of the UAV for detecting the $k$-th target can be respectively expressed as
\begin{align}
& \mathbf{a}_k[n] = \left[
{\mathrm{e}}^{{\mathrm{j}} \frac{2 \pi}{\lambda} x_{1}[n] \sin (\theta_{k}[n]) },  \ldots, 
{\mathrm{e}}^{{\mathrm{j}} \frac{2 \pi}{\lambda} x_{M_{t}}[n] \sin (\theta_{k}[n]) }
\right]^{\mathsf{T}}, \label{transmit FA vector} \\ 
& \mathbf{b}_k[n] = \left[
{\mathrm{e}}^{{\mathrm{j}} \frac{2 \pi}{\lambda} y_{1}[n] \sin (\theta_{k}[n]) },  \ldots, 
{\mathrm{e}}^{{\mathrm{j}} \frac{2 \pi}{\lambda} y_{M_{r}}[n] \sin (\theta_{k}[n]) }
\right]^{\mathsf{T}}, \label{receive FA vector}
\end{align}
where $\lambda$ is the carrier wavelength,
$\theta_{k}[n]$ denotes the vertical angle of departure (AoD) from the UAV to the $k$-target at time slot $n$, and is given by
\begin{equation}
    \theta_{k}[n] = \arcsin \left( \frac{H}{\sqrt{\|\mathbf{q}[n] - \mathbf{q}_{k} \|^{2}+H^2}} \right). \label{theta_m}
\end{equation}
Then, the response matrix of the $k$-th target can be written as
\begin{equation}
    \mathbf{W}_k[n] = \frac{\alpha_{k}}{2d_{k}[n]} \mathbf{b}_k[n] \mathbf{a}_k^{\mathsf{H}}[n], 
\end{equation}
where $\alpha_{k}$ denotes the complex radar cross-section (RCS)~\cite{yuan2021Integrated}, and $d_{k}[n] = \sqrt{\|\mathbf{q}[n] - \mathbf{q}_{k} \|^{2}+H^2}$ represents the distance between the UAV and the $k$-th target.
Accordingly, the received echo signal at the UAV for estimating the $k$-th target can be expressed as
\begin{equation}
    \mathbf{Y}_{k}[n] = \mathbf{W}_k[n] \mathbf{X}_{}[n]  + \mathbf{N}[n],
\end{equation}
where we have
\begin{equation}
    \mathbf{X}_{}[n] = \big[ \mathbf{s}_{}[n;1],\mathbf{s}_{}[n;2],\ldots,\mathbf{s}_{}[n;\bar{N}] \big] \in \mathbb{C}^{M_{t} \times \bar{N}},
\end{equation}
which denotes the transmitted signal from the UAV for target sensing,
and $\mathbf{N}_{}[n] \in \mathbb{C}^{M_{r} \times \Bar{N}}$ represents the additive white Gaussian noise with zero mean and the variance $\sigma_{r}^{2}$ of each entry of $\mathbf{N}_{}[n]$.
Since the number of the transmission frames is very large, we have the following estimation of the covariance of the transmitted signal as
\begin{equation}
  \mathbf{R}_{\mathbf{x}}[n] = \frac{1}{\bar{N}} \mathbf{X}[n] \mathbf{X}^{\mathsf{H}}[n] \approx  \mathbf{R}_{}[n]. \label{covariance_Rx}
\end{equation}
Then, the CRB for estimating the $k$-th target can be given by \eqref{CRB_1}, where $\mathbf{\Psi}_k[n] = \mathbf{b}_k[n] \mathbf{a}^{\mathsf{H}}_k[n]$ and $\tilde{\mathbf{\Psi}}_k[n] = \frac{ \partial \mathbf{\Psi}_k[n]}{ \partial \theta_{k}[n] }$, for simplicity.
According to \cite{liu2022Cramér-Rao}, Eq. \eqref{CRB_1} can be equivalently transformed into Eq. \eqref{CRB_2}, where the detailed transformation is presented in Appendix \ref{crb1to2}.
\begin{figure*}[!t]
\vspace*{-\baselineskip} 
{ \begin{align} 
\mathcal{C}_k[n] &= \frac{ \sigma_{r}^{2} }
{ 2 \left\vert \frac{\alpha_{k}}{2d_{k}[n]} \right\vert^2 \bar{N}
\left( \mathsf{tr}\left( \tilde{\mathbf{\Psi}}_k^{\mathsf{H}}[n] \tilde{\mathbf{\Psi}}_k[n] \mathbf{R}_{\mathbf{x}}[n] \right) - 
\frac{
\left\vert \mathsf{tr}\left( \tilde{\mathbf{\Psi}}_k^{\mathsf{H}}[n] \mathbf{\Psi}_k[n] \mathbf{R}_{\mathbf{x}}[n] \right) \right\vert^2  }{
\mathsf{tr}\left( \mathbf{\Psi}_k^{\mathsf{H}}[n] \mathbf{\Psi}_k[n] \mathbf{R}_{\mathbf{x}}[n] \right)
}   \right)  } \nonumber \\
&= \frac{ 2 d_{k}^{2}[n] \sigma_{r}^{2} \mathsf{tr}\left( \mathbf{\Psi}_k^{\mathsf{H}}[n] \mathbf{\Psi}_k[n] \mathbf{R}_{\mathbf{x}}[n] \right)}
{ \vert \alpha_{k} \vert^2 \bar{N} \left( \mathsf{tr}\left( \tilde{\mathbf{\Psi}}_k^{\mathsf{H}}[n] \tilde{\mathbf{\Psi}}_k[n] \mathbf{R}_{\mathbf{x}}[n] \right) \mathsf{tr}\left( \mathbf{\Psi}_k^{\mathsf{H}}[n] \mathbf{\Psi}_k[n] \mathbf{R}_{\mathbf{x}}[n] \right) - \left\vert \mathsf{tr}\left( \tilde{\mathbf{\Psi}}_k^{\mathsf{H}}[n] \mathbf{\Psi}_k[n] \mathbf{R}_{\mathbf{x}}[n] \right) \right\vert^2 \right)}.
\label{CRB_1}
\end{align} 
\vspace{-\baselineskip}
\begin{align}
\tilde{\mathcal{C}}_k[n] = \frac{ 2 d_{k}^{2}[n] \sigma_{r}^{2} }
 {  \vert \alpha_{k} \vert^2 \bar{N} \left(\frac{2 \pi}{\lambda} \cos (\theta_{k}[n]) \right)^{2} \mathbf{a}_k^{\mathsf{H}}[n] \mathbf{R}_{\mathbf{x}}[n] \mathbf{a}_k[n] \mathbf{y}[n]^{\mathsf{T}} \left( \mathbf{I}_{N_{r}} - \frac{1}{N_{r}} \mathbf{1}_{N_{r}} \mathbf{1}_{N_{r}}^{\mathsf{T}} \right) \mathbf{y}[n] }.  
 \label{CRB_2}
\end{align}
} \hrulefill
\end{figure*}

\subsection{Problem Formulation}
In this work, we aim to minimize the average CRB of all targets during the mission interval, by jointly optimizing the trajectory design of the UAV $\{ \mathbf{q}[n]  \}$, the beamforming $\{ \mathbf{R}_{}[n] \}$, and the antenna positions of the transmit FAs $ \{\mathbf{x}[n] \}$ and the receive FAs $ \{\mathbf{y}[n] \}$.
Hence, the problem can be formulated as
\begin{subequations} 
\begin{flalign}
 (\textbf{P1}):\ & \min_{ \mathbf{q}[n],\mathbf{R}_{}[n],
 \mathbf{x}[n],\mathbf{y}[n] } \quad \frac{1}{N}\frac{1}{K} \sum_{k=1}^{K}\sum_{n=1}^{N} \tilde{\mathcal{C}}_k [n] \nonumber  \\
 {\rm{s.t.}}  \quad &\mathbf{q}[1]= \mathbf{q}_{\rm{I}}, \quad
\mathbf{q}[N]=\mathbf{q}_{\rm{F}}, \label{p1b} \\
 & \| \mathbf{v}[n] \|  \le V_{\rm{max}}, \forall n\in \mathcal{N}\backslash\{1\}, \label{p1c} \\
 & \mathsf{tr}(\mathbf{R}_{}[n]) \le P_{\max},\ \forall n \in \mathcal{N}, \label{p1d}  \\
 & \boldsymbol{\mathcal{X}} \mathbf{x}[n] \succeq \mathbf{I}_{\mathbf{\mathcal{X}}} , \forall n \in \mathcal{N}, \label{p1g} \\
 & \boldsymbol{\mathcal{Y}} \mathbf{y}[n] \succeq \mathbf{I}_{\mathbf{\mathcal{Y}}} , \forall n \in \mathcal{N}, \label{p1h} 
\end{flalign} 
\end{subequations}where
$\mathbf{q}[1]$ and $\mathbf{q}[N]$ are the initial position and final position of the UAV during the mission interval, respectively, as presented in constraint \eqref{p1b}.
Meanwhile, the velocity of the UAV is constrained by the maximum achievable velocity $V_{\max}$ as indicated in constraint \eqref{p1c}.
Constraint \eqref{p1d} limits the transmit power of the UAV that cannot exceed the maximum allowable power $P_{\max}$.
We utilize $\boldsymbol{\mathcal{X}}$ and $\mathbf{I}_{\mathbf{\mathcal{X}}}$ to denote the feasible region constraint of the transmit FA array as illustrated in \eqref{p1g}, i.e., all FAs should move within a region $[0,\mathcal{D}] \triangleq \mathsf{D}$, and any two FAs should keep a minimum distance $\mathcal{D}_{\min}$ to avoid antenna coupling, where $\boldsymbol{\mathcal{X}}$ and $\mathbf{I}_{\mathbf{\mathcal{X}}}$ are given by
\begin{equation}
\boldsymbol{\mathcal{X}} = 
\begin{bmatrix}
-1 & 1 & 0 & 0 & \cdots & 0 & 0 \\
0 & -1 & 1 & 0 & \cdots & 0 & 0 \\
\vdots & \vdots & \vdots & \vdots & \ddots & \vdots & \vdots \\
0 & 0 & 0 & 0 & \cdots & -1 & 1 \\
1 & 0 & 0 & 0 & \cdots & 0 & 0 \\
0 & 0 & 0 & 0 & \cdots & 0 &  -1
\end{bmatrix}_{(M_{t}+1) \times M_{t}},
\end{equation}
\begin{equation}
\mathbf{I}_{\mathcal{X}} = [\mathcal{D}_{\rm{min}}, \mathcal{D}_{\rm{min}}, \cdots, \mathcal{D}_{\rm{min}}, 0, -\mathcal{D}_{\rm{}}]^{\mathsf{T}} \in \mathbb{R}^{(M_{t} + 1) \times 1}.
\end{equation}
Similarly, constraint \eqref{p1h} ensures the feasible region of the receive FAs, where $\boldsymbol{\mathcal{Y}}$ and $\mathbf{I}_{\mathbf{\mathcal{Y}}}$ are given by
\begin{equation}
\boldsymbol{\mathcal{Y}} = 
\begin{bmatrix}
-1 & 1 & 0 & 0 & \cdots & 0 & 0 \\
0 & -1 & 1 & 0 & \cdots & 0 & 0 \\
\vdots & \vdots & \vdots & \vdots & \ddots & \vdots & \vdots \\
0 & 0 & 0 & 0 & \cdots & -1 & 1 \\
1 & 0 & 0 & 0 & \cdots & 0 & 0 \\
0 & 0 & 0 & 0 & \cdots & 0 &  -1
\end{bmatrix}_{(M_{r}+1) \times M_{r}},
\end{equation}
\begin{equation}
\mathbf{I}_{\mathcal{Y}} = [\mathcal{D}_{\rm{min}}, \mathcal{D}_{\rm{min}}, \cdots, \mathcal{D}_{\rm{min}}, 0, -\mathcal{D}_{\rm{}}]^{\mathsf{T}} \in \mathbb{R}^{(M_{r} + 1) \times 1}.
\end{equation}

\section{Proposed Solution}
First, due to the optimization variables, e.g., the trajectory of the UAV, the beamforming, and the antenna position, are coupled in the denominator of the fractional objective function, the problem \textbf{P1} becomes challenging to solve, and leads to the non-convexity with increased complexity.
However, minimizing the original objective is equivalent to maximizing its reciprocal.
Hence, the original problem \textbf{P1} can be reformulated as problem \textbf{P2}, which is given by
\begin{subequations} 
\begin{flalign}
 (\textbf{P2}):\ & \max_{ \mathbf{q}[n],\mathbf{R}_{}[n],
 \mathbf{x}[n],\mathbf{y}[n] } \quad \frac{1}{N}\frac{1}{K} \sum_{k=1}^{K}\sum_{n=1}^{N} \frac{1}{\tilde{\mathcal{C}}_k [n]} \nonumber  \\
 &\quad \quad{\rm{s.t.}}  \quad \eqref{p1b}\text{-}\eqref{p1h} \nonumber
\end{flalign} 
\end{subequations}
The reformulated problem \textbf{P2} is still non-convex due to the non-convexity of the objective function.
To address the issue, we decouple the problem \textbf{P2} into three subproblems, and utilize an AO-based algorithm to solve the overall problem.

\subsection{UAV Trajectory Optimization}
Given the beamforming $\{ \mathbf{R}_{}[n] \}$, and the antenna positions of the FAs $ \{\mathbf{x}[n] \}$ and $ \{\mathbf{y}[n] \}$, the original problem \textbf{P1} can be simplified as
\begin{subequations} 
\begin{flalign}
 (\textbf{P3}):\ & \max_{ \mathbf{q}[n]  } \quad \frac{1}{N}\frac{1}{K} \sum_{k=1}^{K}\sum_{n=1}^{N} \frac{1}{\tilde{\mathcal{C}}_k [n]} \nonumber  \\
 {\rm{s.t.}}  \quad &\mathbf{q}[1]= \mathbf{q}_{\rm{I}}, \quad
\mathbf{q}[N]=\mathbf{q}_{\rm{F}},  \\
 & \| \mathbf{v}[n] \|  \le V_{\rm{max}}, \forall n\in \mathcal{N}\backslash\{1\}. 
\end{flalign} 
\end{subequations}
Firstly, the steering vector $\mathbf{a}_k[n]$ is non-convex and non-linear with respect to the trajectory $\mathbf{q}[n]$ of the UAV.
To tackle this difficulty, we approximate $\mathbf{a}_k[n]$ in the $(l+1)$-th iteration utilizing the trajectory obtained in the $(l)$-th iteration \cite{10654366}, which can be expressed by
\begin{equation} 
\mathbf{a}^{(l)}_k[n] = \left[
{\mathrm{e}}^{{\mathrm{j}} \frac{2 \pi}{\lambda} x_{1}[n] \sin \left(\theta_{k}^{(l)}[n]\right) },  \ldots, 
{\mathrm{e}}^{{\mathrm{j}} \frac{2 \pi}{\lambda} x_{N_{t}}[n] \sin \left(\theta_{k}^{(l)}[n]\right) }
\right]^{\mathsf{T}}, 
\end{equation}
where $\theta_{k}^{(l)}[n] = \arcsin \left( \frac{H}{\sqrt{\|\mathbf{q}^{(l)}[n] - \mathbf{q}_{k} \|^{2}+H^2}} \right), k \in \mathcal{K}$.
Then, the approximated CRB value of the $k$-th target is given by
\begin{equation}
    \tilde{\mathcal{C}}_k^{(l)}[n]=\frac{2d_k^2[n]\sigma_r^2}
    {
    \mathbf{\Lambda}_k^{(l)} [n] \mathbf{\Omega}_k[n]
    },
\end{equation}
where
\begin{equation}
    \mathbf{\Lambda}_k^{(l)} [n] = \left(\frac{2 \pi}{\lambda} \cos (\theta_{k}^{(l)}[n]) \right)^{2} \left(\mathbf{a}_k^{(l)}\right)^{\mathsf{H}}[n] \mathbf{R}_{\mathbf{x}}[n] \mathbf{a}_k^{(l)}[n],
\end{equation}
\begin{equation}
    \mathbf{\Omega}_k[n] = \vert\alpha_k\vert^2 \bar{N}
    \mathbf{y}[n]^{\mathsf{T}} \left( \mathbf{I}_{N_{r}} - \frac{1}{N_{r}} \mathbf{1}_{N_{r}} \mathbf{1}_{N_{r}}^{\mathsf{T}} \right) \mathbf{y}[n].
\end{equation}
Then, problem \textbf{P3} can be transformed as problem \textbf{P3}.$l$, which is given by
\begin{subequations} 
\begin{flalign}
 (\textbf{P3}.l):\ & \max_{ \mathbf{q}[n] } \quad \frac{1}{N}\frac{1}{K} \sum_{k=1}^{K}\sum_{n=1}^{N} \frac{1}{\tilde{\mathcal{C}}_k^{(l)} [n]} \nonumber  \\
 {\rm{s.t.}}  \quad &\mathbf{q}[1]= \mathbf{q}_{\rm{I}}, \quad
\mathbf{q}[N]=\mathbf{q}_{\rm{F}},  \\
 & \| \mathbf{v}[n] \|  \le V_{\rm{max}}, \forall n\in \mathcal{N}\backslash\{1\}. 
\end{flalign} 
\end{subequations}
Problem \textbf{P3}.$l$ is a standard convex optimization problem, which can be solved through existing convex optimizer.

\subsection{Transmit Beamforming Optimization}
Given the trajectory $\{ \mathbf{q}_{}[n] \}$, and the antenna positions of the FAs $ \{\mathbf{x}[n] \}$ and $ \{\mathbf{y}[n] \}$, the original problem \textbf{P1} can be reformulated as
\begin{subequations} 
\begin{flalign}
 (\textbf{P4}):\ & \max_{ \mathbf{R}_{}[n] } \quad \frac{1}{N}\frac{1}{K} \sum_{k=1}^{K}\sum_{n=1}^{N} \frac{1}{\tilde{\mathcal{C}}_k [n]} \nonumber  \\
 {\rm{s.t.}}  \quad &
 \mathsf{tr}(\mathbf{R}_{}[n]) \le P_{\max},\ \forall n \in \mathcal{N}.
\end{flalign} 
\end{subequations}
Since the optimization of the beamforming is independent between time slots, the beamforming variables $\mathbf{R}[n]$ can be decoupled across different time slots. Hence, we decompose problem \textbf{P4} into $N$ sub-slot-problems, where the decoupled problem is given by
\begin{subequations} 
\begin{flalign}
 (\textbf{P4}.n):\ & \max_{ \mathbf{R}_{}[n] } \quad \frac{1}{K} \sum_{k=1}^{K} \frac{1}{\tilde{\mathcal{C}}_k [n]} \nonumber  \\
 {\rm{s.t.}}  \quad &
 \mathsf{tr}(\mathbf{R}_{}[n]) \le P_{\max}.
\end{flalign} 
\end{subequations}
Problem \textbf{P4}.$n$ is a standard convex optimization problem, and we can solve it through conventional convex solver.

\subsection{Antenna Position Optimization}
Given the trajectory $\{ \mathbf{q}_{}[n] \}$, and the beamforming $\{ \mathbf{R}_{}[n] \}$, problem \textbf{P1} can be reformulated as
\begin{subequations} 
\begin{flalign}
 (\textbf{P5}):\ & \max_{ 
 \mathbf{x}[n],\mathbf{y}[n] } \quad \frac{1}{N}\frac{1}{K} \sum_{k=1}^{K}\sum_{n=1}^{N} \frac{1}{\tilde{\mathcal{C}}_k [n]}  \nonumber  \\
 {\rm{s.t.}}  \quad 
 & \boldsymbol{\mathcal{X}} \mathbf{x}[n] \succeq \mathbf{I}_{\mathbf{\mathcal{X}}} , \forall n \in \mathcal{N},  \\
 & \boldsymbol{\mathcal{Y}} \mathbf{y}[n] \succeq \mathbf{I}_{\mathbf{\mathcal{Y}}} , \forall n \in \mathcal{N}.
\end{flalign} 
\end{subequations}
The variables $\mathbf{x}[n]$, and $\mathbf{y}[n]$ are temporally decoupled and can be optimized independently per time slot. To reduce computational complexity, we decompose problem \textbf{P5} into $N$ sub-slot-problems. The optimization for the $n$-th time slot is expressed as
\begin{subequations} 
\begin{flalign}
 (\textbf{P5}.n):\ & \max_{ 
 \mathbf{x}[n],\mathbf{y}[n] } \quad \frac{1}{K} \sum_{k=1}^{K} \frac{1}{\tilde{\mathcal{C}}_k [n]}  \nonumber  \\
 {\rm{s.t.}}  \quad 
 & \boldsymbol{\mathcal{X}} \mathbf{x}[n] \succeq \mathbf{I}_{\mathbf{\mathcal{X}}} , \label{p51} \\
 & \boldsymbol{\mathcal{Y}} \mathbf{y}[n] \succeq \mathbf{I}_{\mathbf{\mathcal{Y}}} \label{p52} .
\end{flalign} 
\end{subequations}
Due to the non-convex and high-dimensional nature of the solution space, conventional optimization techniques and exhaustive search are computationally prohibitive. To address this, we employ a particle swarm optimization (PSO) approach, as described in \cite{10741192, 10818453}.
We initialize a swarm of $P$ particles, each representing a candidate FA position solution. The initial particle set is denoted as
\begin{equation}
    \mathbf{P}^{(0)}[n] = \left\{
        \mathbf{p}_1^{(0)}[n],
        \mathbf{p}_2^{(0)}[n], \ldots,
        \mathbf{p}_P^{(0)}[n]
    \right\},
\end{equation}
where the $p$-th particle is given by
\begin{equation}
    \mathbf{p}_p^{(0)}[n] = \left\{
        {x}_1^{(0)}[n],
        \ldots,
        {x}_{M_t}^{(0)}[n],
        {y}_1^{(0)}[n],
        \ldots,
        {y}_{M_r}^{(0)}[n]
    \right\}.
\end{equation}

The initial velocity of each particle is given by
\begin{equation}
    \mathbf{V}^{(0)}[n] = \left\{ \mathbf{v}_1^{(0)}[n], \mathbf{v}_2^{(0)}[n], \ldots, \mathbf{v}_P^{(0)}[n] \right\}.
\end{equation}
Let $\mathbf{p}_p^\text{best}[n]$ denote the best position found by the $p$-th particle, and $\mathbf{p}^\text{global}[n]$ denote the global best position across the particle swarm. The velocity of each particle is updated as
\begin{equation}
\begin{split}
    \mathbf{v}_p^{(t+1)}[n] \leftarrow
    \omega \cdot \mathbf{v}_p^{(t)}[n] &+
    c_1r_1 \left( \mathbf{p}_p^\text{best}[n] -\mathbf{p}_p^{(t)}[n]\right) \\ &+
    c_2r_2 \left( \mathbf{p}^\text{global}[n] -  \mathbf{p}_p^{(t)}[n]\right),
\end{split}
\label{pvel}
\end{equation}
where $t$ is the iteration index, $c_1$ and $c_2$ are acceleration coefficients, and $r_1, r_2 \sim \mathcal{U}(0,1)$ introduce stochasticity. The inertia weight $\omega$ is dynamically adjusted to balance exploration and exploitation, which is expressed as
\begin{equation}
    \omega = \omega_{\max} - 
    \frac{(\omega_{\max}-\omega_{\min}) \cdot t}{T_{\max}},
    \label{inertia_weight}
\end{equation}
where $\omega_{\max}$ and $\omega_{\min}$ are the maximum and minimum values of the inertia weight $\omega$, and
$T_{\max}$ is the maximum number of iterations for optimizing problem \textbf{P5}.$n$.

The position of the $p$-th particle is updated as
\begin{equation}
    \mathbf{p}_{p}^{(t+1)} [n] \leftarrow
    \mathsf{Proj}_{\mathsf{D}}\left(
        \mathbf{p}_{p}^{(t)} [n] +
        \mathbf{v}_p^{(t+1)}[n]
    \right),
    \label{ppos}
\end{equation}
where $\mathsf{Proj}_{\mathsf{D}}(\cdot)$ denotes the projection onto the feasible antenna position region $\mathsf{D}$. The objective function and constraints are jointly considered in the fitness function for the $t$-th iteration of the $p$-th particle, which is given by
\begin{equation}
    \mathcal{F}\left( \mathbf{p}_{p}^{(t)} [n] \right) = 
     \frac{1}{K} \sum_{k=1}^{K} \frac{1}{\tilde{\mathcal{C}}_k [n]} -
     \eta \times \left \vert \mathcal{V}\left( \mathbf{p}_{p}^{(t)} [n] \right) \right \vert,
     \label{PSOfitness}
\end{equation}
where $\eta$ is a large positive factor to ensure the penalty of the solution violation, and $\mathcal{V}\left( \mathbf{p}_{p}^{(t)} [n] \right)$ denotes the set of the FAs violating the constraints \eqref{p51} and \eqref{p52}.
The detailed procedure of the PSO-based algorithm for solving problem \textbf{P5}.$n$ is summarized in Algorithm \ref{algorithmPSO}.

\begin{algorithm}[t] \footnotesize
	\caption{{The PSO-based Algorithm for Solving \textbf{P5}.$n$.}}
		\begin{algorithmic}[1]
			\REQUIRE {An initial feasible solution ${\bf{x}}[n]$ and ${\bf{y}}[n]$;}
			\STATE
			\textbf{Initialize:} The particles swarm with their initial positions $\mathbf{{P}}^{(0)} [n]$ and velocities $\mathbf{{V}}^{(0)} [n]$, the best position of each particle, and the global best position across the particle swarm;
            \STATE
            Evaluate the initial fitness value;
            \FOR{$t$ $=1$ to $T_{\max}$}
                \STATE
                Update the inertia weight via \eqref{inertia_weight};
            \FOR{$p$ $=1$ to $P$}
                \STATE
                Update the velocity of the $p$-th particle via \eqref{pvel};
                \STATE
                Update the position of the $p$-th particle via \eqref{ppos};
                \STATE
                Calculate the fitness value via \eqref{PSOfitness};
                \IF{the current fitness value is greater than the fitness value of the local best position}
                \STATE
                Update the local best position as the current particle position;
                \ENDIF
                \IF{the current fitness value is greater than the fitness value of the global best position}
                \STATE
                Update the global best position as the current particle position;
                \ENDIF
            \ENDFOR
            \ENDFOR
            \ENSURE
			{${\mathbf{x}}^{} [n](\mathbf{{P}}^{(T_{\max})})$ and $\mathbf{y}[n](\mathbf{{P}}^{(T_{\max})})$ determined by the particle swarm $\mathbf{{P}}^{(T_{\max})} [n]$}.
		\end{algorithmic}
        \label{algorithmPSO}
\end{algorithm}

\subsection{Algorithm Analysis}

\begin{figure}[t]
    \centering
    \includegraphics[width=0.825\linewidth]{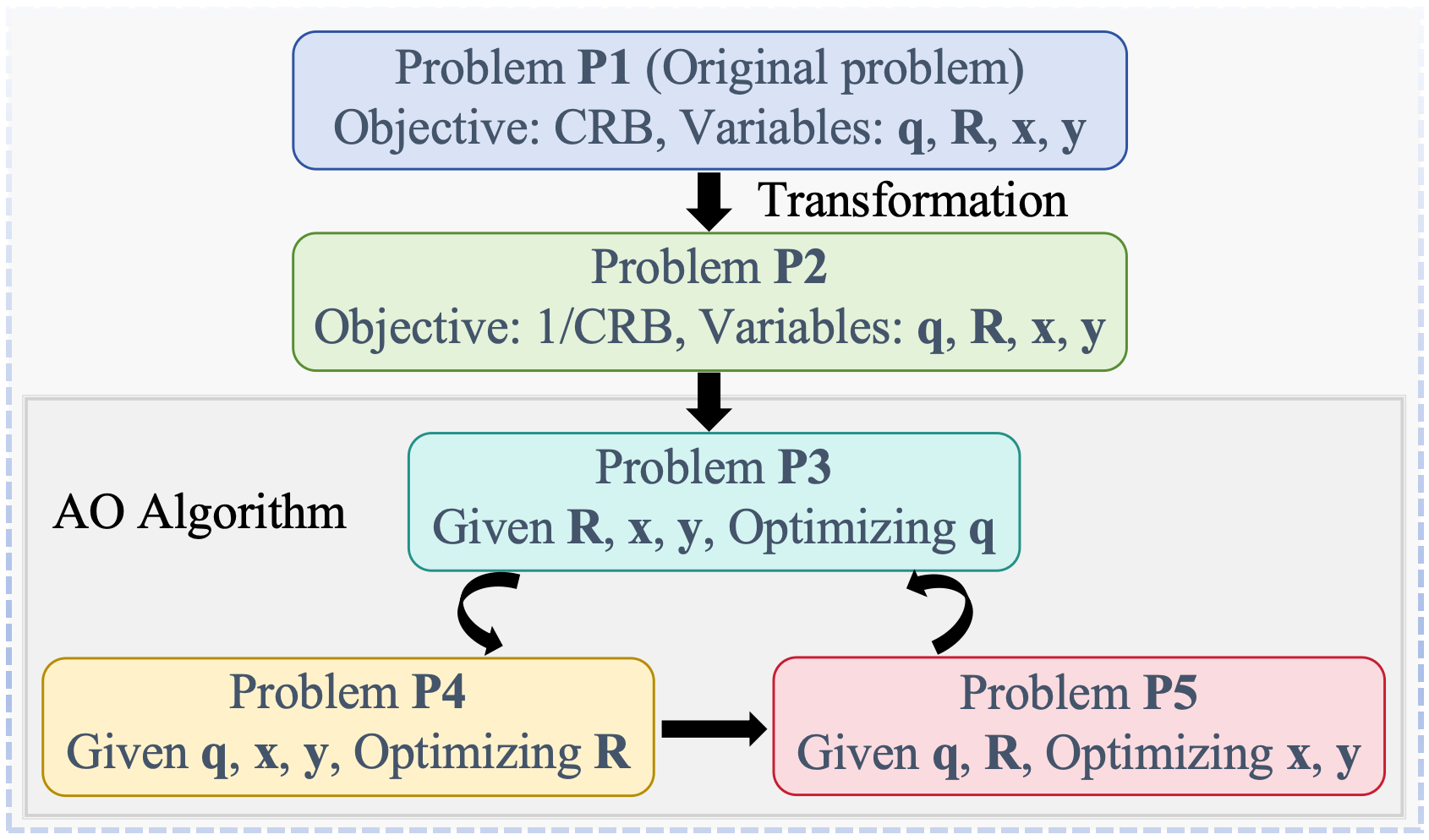}
	\caption{Algorithm procedure of the proposed scheme.}
	\label{algpro}
\end{figure}
The overall procedure of the proposed scheme is illustrated in Fig. \ref{algpro}. Meanwhile, the detailed algorithm of the proposed scheme is summarized in Algorithm~\ref{algorithm1}.
{Then, the following propositions analyze the convergence behavior and computational complexity of Algorithm~\ref{algorithm1}.}
{
\begin{proposition}
    Algorithm~\ref{algorithm1} is guaranteed to converge.
    \label{pconvergence}
\end{proposition}
\begin{proof}
    Please refer to the Appendix \ref{propostionconv}.
\end{proof}
}

{
\begin{proposition}
The worst-case computational complexity of Algorithm~\ref{algorithm1} can be given by
\begin{equation}
\begin{split}
    \digamma_{\mathrm{Alg2}} & = \mathcal{O}\Big( l_{\rm max}\left[ (2N)^{3.5} + 
(NM_{t}^{2})^{3.5} \right]\log \left(\epsilon^{-1}\right) \\
&\quad\quad\ + N t_{\rm{max}} P (M_{t} + M_{r}) \Big).
\end{split}
\label{cca2}
\end{equation}
\end{proposition}
\begin{proof}
    In Algorithm~\ref{algorithm1}, the subproblem \textbf{P3}.($l$) and the subproblem \textbf{P4}.($n$) are solved using the interior-point method \cite{10654366}, with computational complexities $\mathcal{O}\big((2N)^{3.5}\log(\varepsilon^{-1})\big)$ and $\mathcal{O}\big((M_{t}^{2})^{3.5}\log(\varepsilon^{-1})\big)$, respectively. 
Furthermore, the PSO algorithm is employed to solve subproblem \textbf{P5}.($n$), with a computational complexity of $\mathcal{O} ( t_{\rm{max}} P (M_{t} + M_{r})   ) $~\cite{10693833}.
Therefore, the worst-case computational complexity of  Algorithm~\ref{algorithm1} can be derived as \eqref{cca2}.
\end{proof}
}

\begin{algorithm}[t] \footnotesize
\caption{{AO-based Algorithm for Solving Problem (\textbf{P1}).}}
\label{algorithm1} 
\begin{algorithmic}[1]
\REQUIRE
{An initial feasible solution $\{\mathbf{q}^{(0)}[0]\}$, $ \mathbf{R}_{}^{(0)}[n] $, $ \{\mathbf{x}^{(0)}[n] \}$ and $ \{\mathbf{y}^{(0)}[n] \}$, iteration index $ l = 1 $, maximum iteration number $l_{\rm{max}}$, accuracy threshold $ \varepsilon > 0 $, iteration index of the PSO $t = 1$, and maximum iteration number of the PSO $t_{{\max}}$}.
\STATE
\textbf{Repeat:}
\STATE Given $\{\mathbf{R}^{(l)}[n]\}$, $ \{\mathbf{x}^{(l)}[n] \}$, and $ \{\mathbf{y}^{(l)}[n] \}$, solve the problem \textbf{P3}.$l$ to obtain $ \mathbf{q}_{}^{(l+1)}[n]$.
\FOR{{$n$ $=1$ to $N$}}
\STATE Given $\{\mathbf{q}^{(l+1)}[n]\}$, $  \mathbf{x}_{}^{(l)}[n] $, and $ \{\mathbf{y}^{(l)}[n] \}$, solve the problem \textbf{P4}.$n$ to obtain $ \{\mathbf{R}^{(l+1)}[n] \}$.
\STATE Given $\{\mathbf{q}^{(l+1)}[n]\}$ and $ \{\mathbf{R}^{(l+1)}[n] \}$,  update $ \{\mathbf{x}^{(l+1)}[n] \}$ and $ \{\mathbf{y}^{(l+1)}[n] \}$ according to Algorithm \ref{algorithmPSO}.
\ENDFOR
\STATE Update the objective function value of problem \textbf{P2}.
\STATE Update $ l = l + 1 $.
\STATE
\textbf{Until:} the increase of the value of the objective function between two
adjacent iterations is smaller than $ \varepsilon $ or $l > l_{\max}$.
\ENSURE
{The optimized solution for the UAV trajectory, the transmit beamforming, and the antenna position.}
\end{algorithmic} 
\end{algorithm}

\section{Numerical Results}
In this section, we present the numerical results to demonstrate the performance of the proposed UAV-enabled FAS for multi-target wireless sensing over the LAWCNs.
As illustrated in Fig. \ref{figs1}, the UAV flies over an $800 \mathrm{m} \times 800 \mathrm{m}$ LAE mission region to sense $K = 6$ targets at an altitude $H = 100 \mathrm{m}$. The maximum velocity of the UAV is $20 \mathrm{m/s}$. The mission period $T = 45 $ seconds. The number of time slots is $N = 20$. The number of transmission time frames is $\bar{N} = 200$. The number of the transmit FAs at the UAV is $N_t = 12$, and the number of the receive FAs at the UAV is $N_r = 12$. The noise $\sigma_r^2$ is $-70 \mathrm{dBm}$. The RCS of each target is $1\times 10^{-6} $. The feasible antenna region of the FAs is $[0, 20\lambda]$, where the carrying wavelength $\lambda = 0.0107 \mathrm{m}$, and the minimum distance between any two FAs is $\mathcal{D}_{\min} = 0.5 \lambda$. Besides, in the PSO-based algorithm, the maximum iteration number $T_{\max} = 50$, the number of the particles is $P = 50$, the acceleration coefficients are $c_1 = 1.5$ and $c_2 = 1.5$, and the maximum and minimum values of the inertia weight are $\omega_{\max} = 0.9$ and $\omega_{\min} = 0.4$.

To verify the effectiveness of our proposed scheme, the following three benchmark schemes are compared.
\begin{itemize}
    \item Transmit FA optimization (TFAO) scheme: In this scheme, only the transmit antennas are the FAs and the position of the transmit FAs are optimized, while the receive antennas are uniform linear array.
    \item Sparse uniform linear array (SULA) scheme: In this scheme, the transmit antennas and receive antennas are configured as two uniform linear arrays, where the distance between two adjacent antennas has the largest configurable distance, i.e., $\mathcal{D}/{(N_t-1)}$ for the transmit antennas and $\mathcal{D}/{(N_t-1)}$ for the receive antennas, respectively.
    \item Dense uniform linear array (DULA) scheme: In this scheme, the transmit antennas and receive antennas are also configured as two uniform linear arrays, where the distance between two adjacent antennas is configured as the half of the carrying wavelength.
\end{itemize}

\begin{figure}[t]
    \centering
    \includegraphics[width=0.78\linewidth]{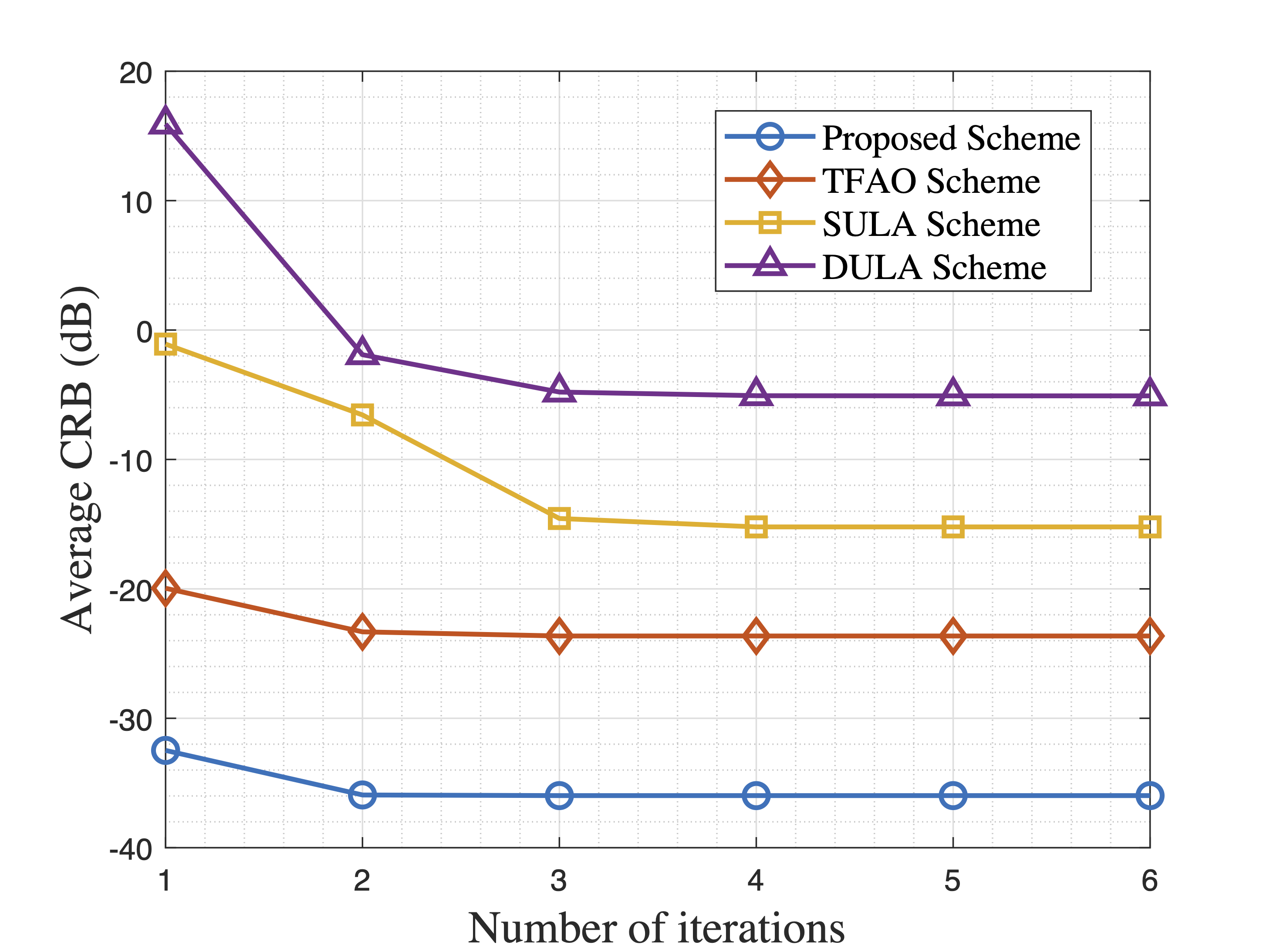}
	\caption{Convergence behavior of different schemes.}
	\label{figs1}
\end{figure}

{
Fig. \ref{figs1} illustrates the convergence behavior of the proposed scheme and the benchmark schemes.
As we can observe from Fig. \ref{figs1}, the proposed scheme achieves the fastest convergence rate. The reason is that, our proposed scheme is capable of fully leveraging the spatial exploration capability of the FAs, which contributes to better channel quality improvement and thereby enhances the overall CRB value. Furthermore, the proposed scheme also converges to the lowest average CRB performance, owing to the joint optimization of the trajectory, the beamforming design, and the FA positioning.
}

\begin{figure}[t]
    \centering
    \includegraphics[width=0.78\linewidth]{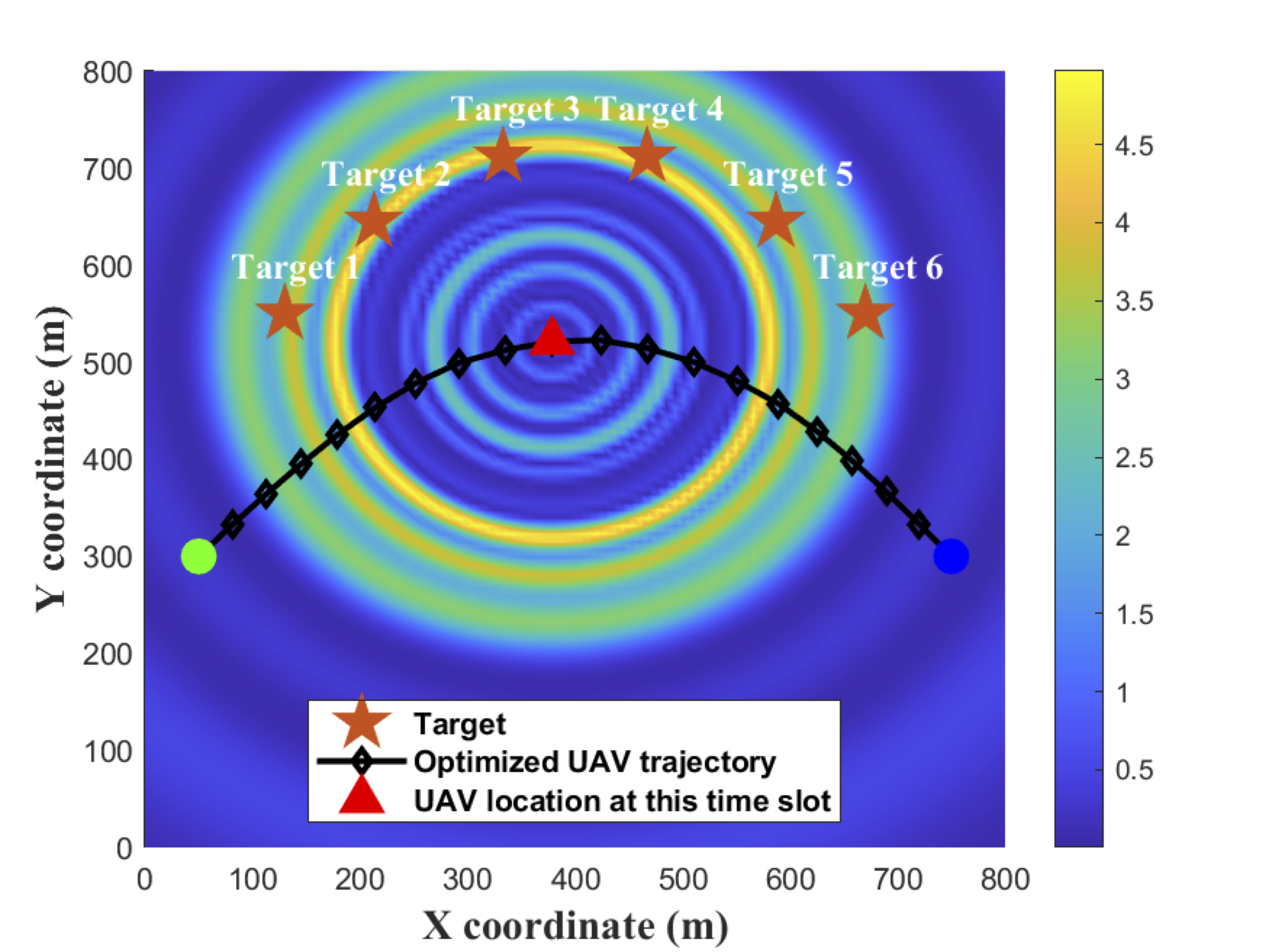}
	\caption{Achievable beampattern gains at time slot $10$.}
	\label{figs2}
\end{figure}

Fig. \ref{figs2} presents the optimized UAV trajectory during the mission interval, and the location of each target. Meanwhile, the achievable beampattern gains toward targets at time slot $10$ is also presented, where the position of the UAV at time slot $10$ is marked in a red triangle.
As we can observe, since the trajectory of the UAV is jointly optimized in our proposed scheme, the UAV deliberately approaches the targets during its flight from the initial position to the final position, thereby achieving higher channel quality and improved sensing performance.
Furthermore, at time slot $10$, the joint optimization of the beamforming and the antenna positions of the transmit and receive FAs takes into account the current position of the UAV as well as the positions of all targets. As a result, the UAV directs more beams toward the region where the targets are located, enhancing the beampattern gain. Additionally, Fig. \ref{figs3} illustrates the achievable CRB performance for a selected set of targets, i.e., the target $1$, $3$, and $5$ illustrated in Fig. \ref{figs2}. It can be observed that the CRB values for these targets decrease to varying degrees when the UAV is closer to each of them, owing to the improved channel conditions resulting from the reduced distance.

\begin{figure}[t]
    \centering
    \includegraphics[width=0.78\linewidth]{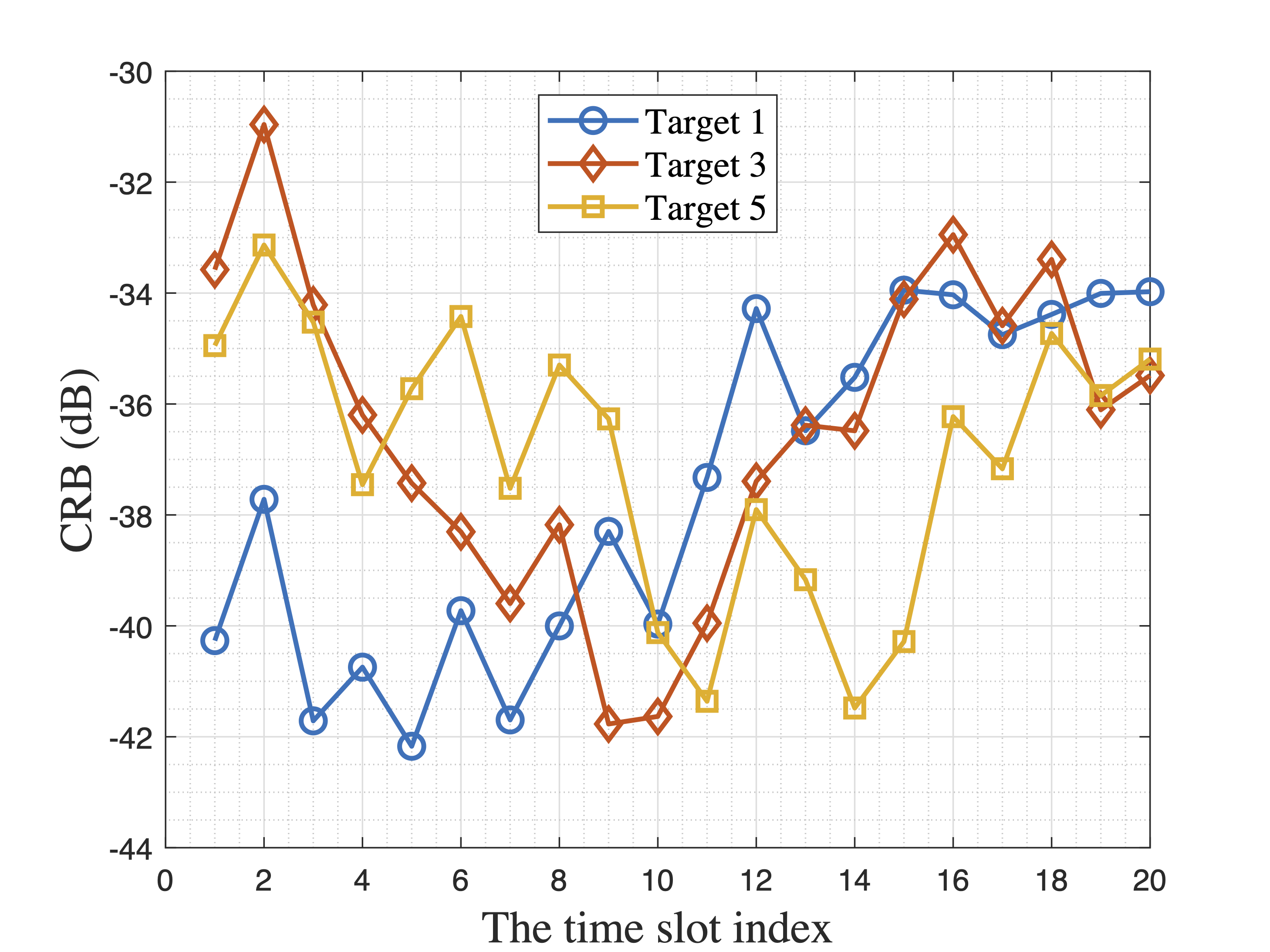}
	\caption{Achievable beampattern gains of selected targets.}
	\label{figs3}
\end{figure}

\begin{figure}[t]
    \centering
    \includegraphics[width=0.78\linewidth]{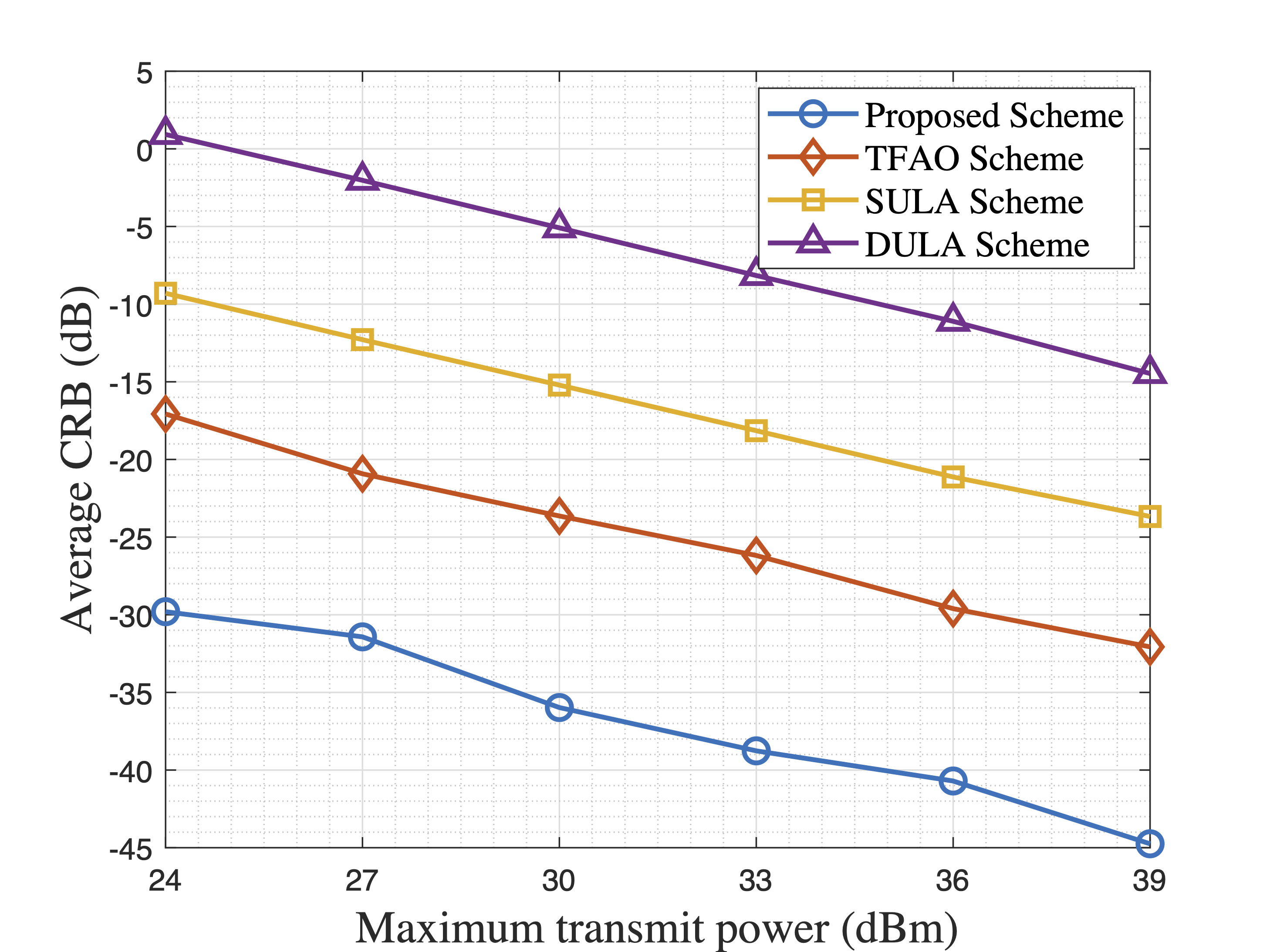}
	\caption{Average CRB value versus maximum transmit power under different schemes.}
	\label{figs4}
\end{figure}

{
Fig. \ref{figs4} demonstrates the average CRB value versus the maximum transmit power under different schemes.
It can be observed that as the maximum transmit power increases, the CRB values of all schemes decrease accordingly. The reason is that, the higher transmit power enables greater beampattern gain, which in turn reduces the CRB value. Among these schemes, the proposed scheme achieves the lowest CRB performance across the entire maximum transmit power range. The improvement can be attributed to the joint optimization design, which leverages a fully utilization of the feasible antenna movement range of the transmit and receive FAs compared to other schemes.
}

\begin{figure}[t]
    \centering
    \includegraphics[width=0.78\linewidth]{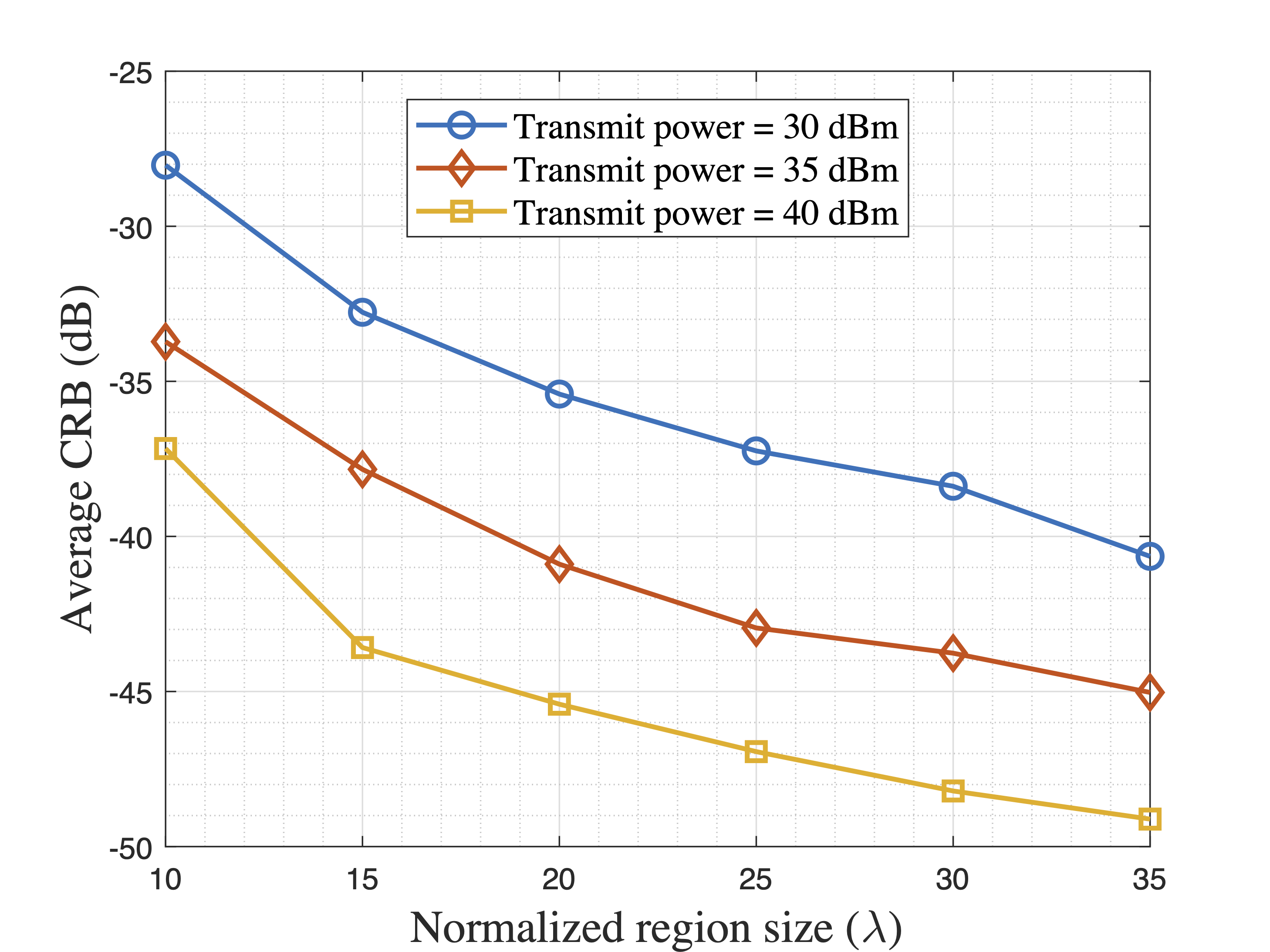}
	\caption{Average CRB value versus normalized region size.}
	\label{figs5}
\end{figure}

{
As illustrated in Fig. \ref{figs5}, the average CRB value versus the normalized region size of the proposed scheme is compared under different transmit powers, where the normalized region size corresponds to the maximum reconfigurable feasible region of the FAs.
As we can observe, with the normalized region size expands, the achievable average CRB values of all targets decrease under different transmit powers of the UAV. The reason is that, a larger antenna reconfigurable range provides more DoFs for the transmit and receive FAs, enabling more precise and flexible beam steering towards the targets, which improves channel quality and thereby enhances the CRB accuracy.
}

\begin{figure}[t]
    \centering
    \includegraphics[width=0.78\linewidth]{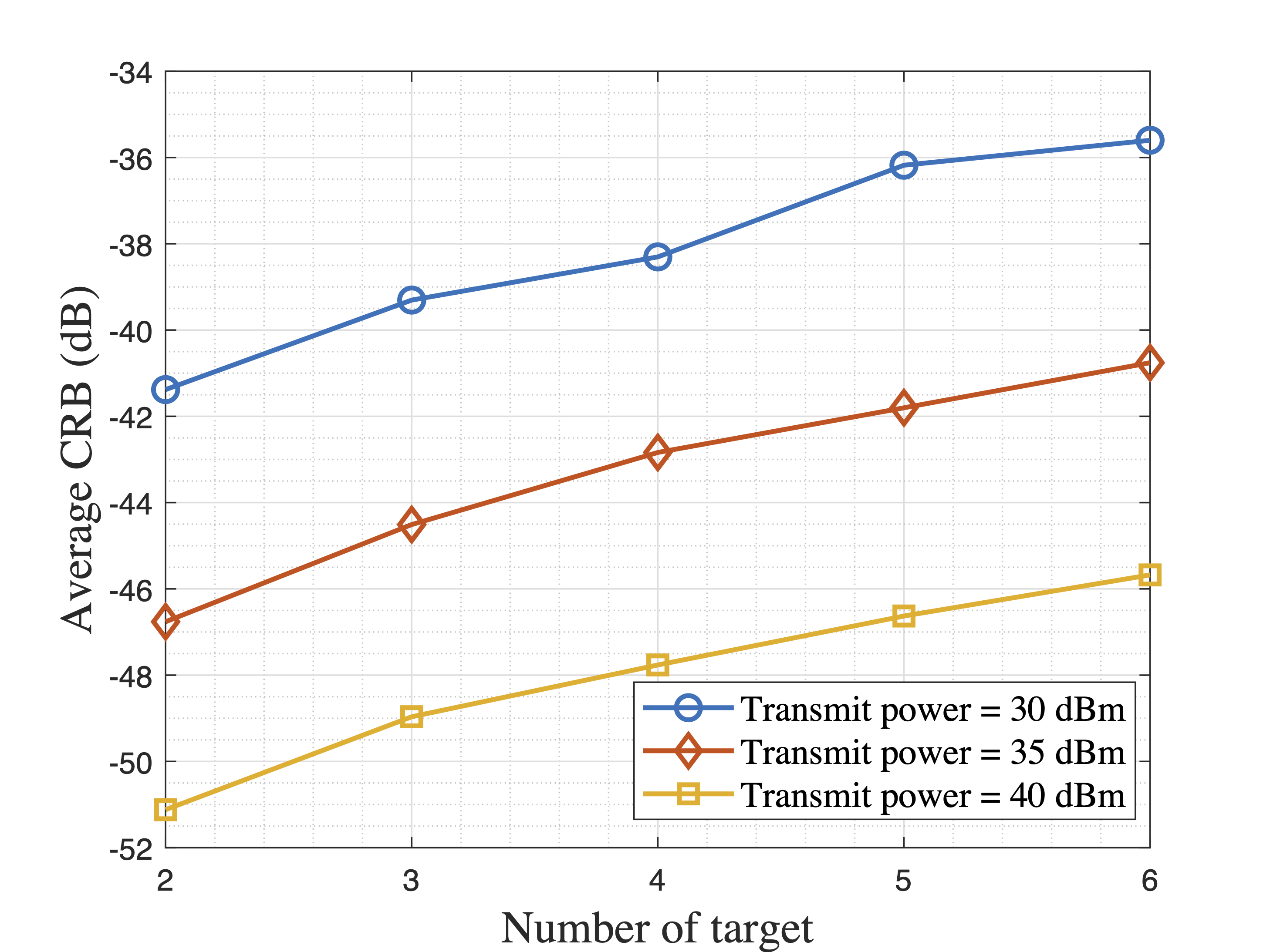}
	\caption{Average CRB value versus the number of targets.}
	\label{figs6}
\end{figure}
{
In Fig. \ref{figs6}, the average CRB value of the targets versus the number of the targets of the proposed scheme is presented.
It can be observed that as the number of the targets increases, the average CRB value also rises. The reason is that, a larger number of targets leads to a more dispersed distribution of the beamforming energy, thereby reducing the beampattern gain per target and resulting in the degraded average CRB performance. On the other hand, when the UAV transmits with a higher transmit power, the CRB performance improves across all considered target counts. The trend is consistent with our previous analysis for Fig. \ref{figs4}.
}


\section{Conclusion}
This paper demonstrates the significant potential of the UAV-enabled FAS in improving multi-target sensing performance through the joint UAV trajectory design, antenna position optimization of the transmit FAs and the receive FAs, and the UAV beamforming optimization. The proposed AO algorithm effectively addresses the non-convex CRB minimization problem and achieves superior estimation accuracy and operational reliability, which are validated in the numerical results.

{
Looking ahead, future work can extend this framework to incorporate integrated sensing and communication techniques, with the goal of enabling simultaneous target sensing and high-speed ubiquitous data transmission within the LAWCNs.
Further investigations can also explore the coordination of the UAV-enabled sensing with other aerial services, such as low-altitude control and ultra-reliable computation services, to achieve efficient resource utilization in next-generation aerial platforms for performing the LAE-oriented consumer missions efficiently and reliably.
}

\begin{appendices}

\section{Transformation of Eq. \eqref{CRB_1}}
\label{crb1to2}
Similar to \cite{liu2022Cramér-Rao}, we first give the derivatives of the steering vectors for detecting the $k$-th target with respect to the vertical AoD from the UAV to the $k$-target at time slot $n$, which are given by
\begin{equation}
\begin{split}
    \tilde{\mathbf{a}}_k[n] =& \left[{\text{j}} \frac{2 \pi}{\lambda}x_1[n]\cos(\theta_k[n]),
    \ldots,\right.\\&\ \ \left.
    {\text{j}} \frac{2 \pi}{\lambda}x_{N_t}[n]\cos(\theta_k[n])\right]^{\mathsf{T}} \odot \mathbf{a}_k[n], 
\end{split}
\end{equation}
\begin{equation}
\begin{split}
    \tilde{\mathbf{b}}_k[n] =& \left[{\text{j}} \frac{2 \pi}{\lambda}y_1[n]\cos(\theta_k[n]),
    \ldots,\right.\\&\ \ \left.
    {\text{j}} \frac{2 \pi}{\lambda}y_{N_r}[n]\cos(\theta_k[n])\right]^{\mathsf{T}} \odot \mathbf{b}_k[n].
\end{split}
\end{equation}
Accordingly, we have
\begin{equation}
\begin{split}
    \mathsf{tr}\left( \mathbf{\Psi}_k^{\mathsf{H}}[n] \mathbf{\Psi}_k[n] \mathbf{R}_{\mathbf{x}}[n] \right) & = \mathsf{tr}\left(\mathbf{a}_k[n] \mathbf{b}_k^{\mathsf{H}}[n] \mathbf{b}_k[n] \mathbf{a}_k^{\mathsf{H}}[n] \mathbf{R}_{\mathbf{x}}[n] \right)   \\
    & = N_{r} \mathbf{a}_k^{\mathsf{H}}[n] \mathbf{R}_{\mathbf{x}}[n] \mathbf{a}_k[n] = N_{r}\mathbf{A}_k[n],
    \label{appcrba1}
\end{split}
\end{equation}
\begin{equation}
\begin{split}
    &\mathsf{tr}\left(\tilde{\mathbf{\Psi}}_k^{\mathsf{H}}[n] \mathbf{\Psi}_k[n] \mathbf{R}_{\mathbf{x}}[n]\right)\\&= 
    \mathsf{tr}\left( \left( \mathbf{a}_k[n] \tilde{\mathbf{b}}_k^{\mathsf{H}}[n] + \tilde{\mathbf{a}}_k[n] \mathbf{b}_k^{\mathsf{H}}[n] \right) \mathbf{b}_k[n] \mathbf{a}_k^{\mathsf{H}}[n] \mathbf{R}_{\mathbf{x}}[n] \right)\\& 
    = \kappa_k[n] \mathbf{A}_k[n] \sum_{i = 1}^{N_{r}} y_{i}[n] + N_{r} \mathbf{a}_k^{\mathsf{H}}[n] \mathbf{R}_{\mathbf{x}}[n] \tilde{\mathbf{a}}_k[n],
    \label{appcrba2}
\end{split}
\end{equation}
\begin{equation}
\begin{split}
    &\mathsf{tr}\left(\tilde{\mathbf{\Psi}}_k^{\mathsf{H}}[n] \tilde{\mathbf{\Psi}}_k[n] \mathbf{R}_{\mathbf{x}}[n]\right)  \\ & = 
     \mathsf{tr}\left( \left( \mathbf{a}_k[n] \tilde{\mathbf{b}}_k^{\mathsf{H}}[n] + \tilde{\mathbf{a}}_k[n] \mathbf{b}_k^{\mathsf{H}}[n] \right) \cdot \right. \\& \quad \  \left.
     \left( \tilde{\mathbf{b}}_k[n] \mathbf{a}_k^{\mathsf{H}}[n] + \mathbf{b}_k[n] \tilde{\mathbf{a}}_k^{\mathsf{H}}[n] \right) \mathbf{R}_{\mathbf{x}}[n]\right)\\ \ \ 
     & =  -\left(\kappa_k[n] \right)^{2} \mathbf{A}_k[n] \sum_{i = 1}^{N_{r}} y_{i}^{2}[n] + N_{r} \tilde{\mathbf{a}}_k^{\mathsf{H}}[n] \mathbf{R}_{\mathbf{x}}[n] \tilde{\mathbf{a}}_k[n] \\
     & \ \ \ \ \, - \kappa_k[n] \tilde{\mathbf{A}}_k[n] \sum_{i = 1}^{N_{r}} y_{i}[n],
     \label{appcrba3}
\end{split}
\end{equation}
\begin{equation}
\begin{split}
    &\left| \mathsf{tr}\left( \tilde{\mathbf{\Psi}}_k^{\mathsf{H}}[n] \mathbf{\Psi}_k[n] \mathbf{R}_{\mathbf{x}}[n] \right) \right|^2 \\&= \left( \kappa_k[n] \mathbf{A}_k[n] \sum_{i = 1}^{N_{r}} y_{i}[n]  +  N_{r} \mathbf{a}_k^{\mathsf{H}}[n] \mathbf{R}_{\mathbf{x}}[n] \tilde{\mathbf{a}}_k[n] \right)\\
    &\quad\ \times \left( -\kappa_k[n] \mathbf{A}_k[n] \sum_{i = 1}^{N_{r}} y_{i}[n] + N_{r} \tilde{\mathbf{a}}_k^{\mathsf{H}}[n] \mathbf{R}_{\mathbf{x}}[n] \mathbf{a}_k[n] \right)\\
    &= - (\kappa_k[n]\mathbf{A}_k[n])^2 \left( \sum_{i = 1}^{N_{r}} y_{i}[n] \right)^{2} \\
    &\quad\ + N_{r}^{2} \mathbf{a}_k^{\mathsf{H}}[n] \mathbf{R}_{\mathbf{x}}[n] \tilde{\mathbf{a}}_k[n] \tilde{\mathbf{a}}_k^{\mathsf{H}}[n] \mathbf{R}_{\mathbf{x}}[n] \mathbf{a}_k[n]\\
    &\quad\ - \kappa_k[n] N_{r} \mathbf{A}_k[n]
\tilde{\mathbf{A}}_k[n] \sum_{i = 1}^{N_{r}} y_{i}[n], 
\label{appcrba4}
\end{split}
\end{equation}
where
\begin{equation}
    \mathbf{A}_k[n] = \mathbf{a}_k^{\mathsf{H}}[n] \mathbf{R}_{\mathbf{x}}[n] \mathbf{a}_k[n],
\end{equation}
\begin{equation}
    \kappa_k[n] = - {\text{j}} \frac{2 \pi}{\lambda} \cos (\theta_{k}[n]),
\end{equation}
\begin{equation}
    \tilde{\mathbf{A}}_k[n] = \mathbf{a}_k^{\mathsf{H}}[n] \mathbf{R}_{\mathbf{x}}[n] \tilde{\mathbf{a}}_k[n] - \tilde{\mathbf{a}}_k^{\mathsf{H}}[n] \mathbf{R}_{\mathbf{x}}[n] \mathbf{a}_k[n].
\end{equation}
By substituting Eqs. \eqref{appcrba1}, \eqref{appcrba2}, \eqref{appcrba3}, and \eqref{appcrba4} into Eq. \eqref{CRB_1}, we can thus obtain Eq. \eqref{CRB_2}.

{
\section{Proof of the Proposition \ref{pconvergence}}
\label{propostionconv}
\begin{proof}
We denote ${\mathbf{q}}^{l}$,$ {\mathbf{R}}^{l}$,$ {\mathbf{x}}^{l}$, and $ {\mathbf{y}}^{l}$ as the solution obtained in the $l$-th iteration by Algorithm \ref{algorithm1}, and ${{\mathcal{C}}}\left({\mathbf{q}}^{l}, {\mathbf{R}}^{l}, {\mathbf{x}}^{l}, {\mathbf{y}}^{l}\right)$ as the value of the objective function of problem \textbf{P2}.
In the procedure of Algorithm \ref{algorithm1}, we firstly solve subproblem \textbf{P3}.$l$ to obtain the solution of the UAV trajectory $ {\mathbf{q}}^{l}$ when
given the beamforming $ {\mathbf{R}}^{l}$, and the antenna positions $ {\bf{x}}^{l}$, and $ {\bf{y}}^{l}$. Accordingly, we have
\begin{equation}
     {{\mathcal{C}}}\left({\bf{q}}^{l+1}, {\bf{R}}^{l}, {\bf{x}}^{l}, {\bf{y}}^{l}\right)
\geq {{\mathcal{C}}}\left({\bf{q}}^{l}, {\bf{R}}^{l}, {\bf{x}}^{l}, {\bf{y}}^{l}\right).
\label{pps1}
\end{equation}
Then, given the UAV trajectory $ {\bf{q}}^{l+1}$ and the antenna positions $ {\bf{x}}^{l}$ and $ {\bf{y}}^{l}$, we solve the beamforming optimization subproblems \textbf{P4}.$n$ of all time slots to obtain $\mathbf{R}^{l+1}$.
Then, we have
\begin{equation}
     {{\mathcal{C}}}\left({\bf{q}}^{l+1}, {\bf{R}}^{l+1}, {\bf{x}}^{l}, {\bf{y}}^{l}\right)
\geq {{\mathcal{C}}}\left({\bf{q}}^{l+1}, {\bf{R}}^{l}, {\bf{x}}^{l}, {\bf{y}}^{l}\right).
\label{pps2}
\end{equation}
Similarly, given the UAV trajectory $ {\bf{q}}^{l+1}$ and beamforming $ {\bf{R}}^{l+1}$, we solve the antenna position optimization subproblems \textbf{P5}.$n$ of all time slots by Algorithm \ref{algorithmPSO} to obtain $\mathbf{x}^{l+1}$ and $\mathbf{y}^{l+1}$.
Then, we have
\begin{equation}
     {{\mathcal{C}}}\left({\bf{q}}^{l+1}, {\bf{R}}^{l+1}, {\bf{x}}^{l+1}, {\bf{y}}^{l+1}\right)
\geq {{\mathcal{C}}}\left({\bf{q}}^{l+1}, {\bf{R}}^{l+1}, {\bf{x}}^{l}, {\bf{y}}^{l}\right).
\label{pps3}
\end{equation}
By combining \eqref{pps1}, \eqref{pps2}, and \eqref{pps3}, we can obtain
\begin{equation}
     {{\mathcal{C}}}\left({\bf{q}}^{l+1}, {\bf{R}}^{l+1}, {\bf{x}}^{l+1}, {\bf{y}}^{l+1}\right)
\geq {{\mathcal{C}}}\left({\bf{q}}^{l}, {\bf{R}}^{l}, {\bf{x}}^{l}, {\bf{y}}^{l}\right),
\end{equation}
which indicates that the value of the objective function in problem \textbf{P2} is non-decreasing over iterations.
Meanwhile, the achievable CRB for all targets is bounded by the resource constraints, e.g., the transmit power of the UAV, the trajectory design of the UAV, and the feasible region of the FAs.
Hence, Algorithm \ref{algorithm1} is guaranteed to converge.
\end{proof}
}

\end{appendices}

\bibliographystyle{IEEEtran}
\bibliography{ref}
\end{document}